\documentclass[10pt,onecolumn]{IEEEtran}
\pdfoutput=1
 
\IEEEoverridecommandlockouts

\usepackage[fleqn]{amsmath}
\usepackage{amssymb,graphicx,subfig,verbatim,enumitem,cite}
\usepackage[colorlinks = true, linkcolor = blue, citecolor = blue]{hyperref}

\newtheorem{theorem}{Theorem}

\newtheorem{lemma}{Lemma}
\newtheorem{remark}{Remark}

\providecommand{\norm}[1]{\lVert#1\rVert}
\DeclareMathOperator*{\argmin}{arg\,min}

\title{\LARGE \bf On Event Triggered Tracking for Nonlinear Systems}

\author{Pavankumar Tallapragada and Nikhil Chopra
\thanks{This work was partially supported by Minta Martin Fund and by the National Science Foundation through grant number 0931661.}
\thanks{Pavankumar Tallapragada is with the Department of Mechanical Engineering,
        University of Maryland, College Park, 20742 MD, USA
        {\tt\small pavant@umd.edu}}%
\thanks{N. Chopra is with the Department of Mechanical Engineering and The Institute for Systems Research,
        University of Maryland, College Park, 20742 MD, USA
        {\tt\small nchopra@umd.edu}}%
}

\begin{document}
\maketitle

\begin{abstract}
In this paper we study an event based control algorithm for trajectory tracking in nonlinear systems. The desired trajectory is modelled as the solution of a reference system with an exogenous input and it is assumed that the desired trajectory and the exogenous input to the reference system are uniformly bounded. Given a continuous-time control law that guarantees global uniform asymptotic tracking of the desired trajectory, our algorithm provides an event based controller that not only guarantees uniform ultimate boundedness of the tracking error, but also ensures non-accumulation of inter-execution times. In the case that the derivative of the exogenous input to the reference system is also uniformly bounded, an arbitrarily small ultimate bound can be designed. If the exogenous input to the reference system is piecewise continuous and not differentiable everywhere then the achievable ultimate bound is constrained and the result is local, though with a known region of attraction. The main ideas in the paper are illustrated through simulations of trajectory tracking by a nonlinear system.
\end{abstract}

\section{Introduction}

Traditional computer based control systems rely on periodic sampling of the sensors and computation/execution of the control. The reason for the popularity of this paradigm is a well developed theory and the ease of analysis of such systems. However, such control algorithms may be very inefficient from a computational perspective as the period for sampling and control execution is determined by a worst case analysis and the rate of control execution is independent of the system's state. On the other hand, in event based control systems, timing of control execution is not necessarily periodic and can be state dependent. Thus, event based control is useful in systematically designing controllers that make better use of computational and communication resources in a wide variety of applications such as embedded control systems and decentralized systems (a representative list of references includes \cite{tabuada2007, mazo2011, wang2008, wang2011, dimarogonas2012}).

While there have been some efforts in the past to study event based control systems \cite{arzen1999,astrom2002,astrom2003}, their systematic design for tasks such as stabilization has been undertaken only recently \cite{sandee_phd, tabuada2007, heemels2008, wang2009, wang2010}. Of these, \cite{tabuada2007} has significantly influenced the proposed controller in this paper. In \cite{tabuada2007}, an event-triggering algorithm was proposed that ensures global asymptotic stability as well as a lower bound on the inter-execution times of the control law for general nonlinear systems that are rendered Input-to-State Stable (ISS) with respect to measurement errors by a continuous time controller.

In this paper, we investigate an event triggered control algorithm for trajectory tracking. Tracking a time varying trajectory or even a set-point is of tremendous practical importance in many control applications. In these applications, the goal is to make the state of the system follow a reference or desired trajectory, which is usually specified as an exogenous input to the system. In this paper, the reference trajectory is generated by a reference system. To the best of our knowledge, the majority of the previous works in the event-triggered control literature assumed a state feedback control strategy with no exogenous input, some exceptions being \cite{sandee_phd, heemels2008, wang2009, wang2010, lunze2010, donkers2010, yu2011}, where unknown disturbances appear as exogenous inputs. However, in this paper, we consider exogenous inputs that are available to the controller through measurements, namely the reference trajectory and the input to the reference system.

The \textbf{main contribution} of this paper is the design of event-triggered controllers for trajectory tracking in nonlinear systems, which is a special case of nonlinear systems with exogenous inputs. It is assumed that the reference trajectory and the exogenous input to the reference system are uniformly bounded. Given a nonlinear system and a continuous-time controller that ensures global uniform asymptotic tracking of the desired trajectory, the proposed algorithm provides an event based controller that guarantees uniform ultimate boundedness of the tracking error and ensures that the inter-execution times of the controller are bounded away from zero. In the special case that the derivative of the exogenous input to the reference system is also uniformly bounded, an arbitrarily small ultimate bound for the tracking error can be designed. In this paper, unlike in the event-triggered control literature, the continuous-time control law is assumed to render the closed loop system asymptotically stable rather than ISS with respect to measurement errors. Although on compact sets the latter condition can be arrived at from the former, our choice allows a direct and clear procedure for designing an event-triggering condition with time-varying components that results in fewer controller executions. A preliminary version of the results in this paper have been published in \cite{pavan2011}. The results therein have been expanded here.

The rest of the paper is organized as follows. In Section \ref{sec:prob_stat} we set up the problem and introduce the notation used in the paper. Subsequently, in Section \ref{sec:trig_con}, the major assumptions are stated and the event triggering condition is introduced. The main analytical results are presented in Section \ref{sec:ultim_track}. The theoretical results in the paper are illustrated through numerical simulations of a second order nonlinear system in Section \ref{sec:sim}. Finally, the results are summarized in Section \ref{sec:conc}.


\section{Problem statement and notation}
\label{sec:prob_stat}

Consider a nonlinear system of the form
\begin{equation}
\dot{x} = f(x,u), \quad x \in \mathbb{R}^n, \,\, u \in \mathbb{R}^m
\label{eqn:x_sys}
\end{equation}
which has to track a reference trajectory defined implicitly by the dynamical system
\begin{align}
\dot{x}_d = f_r(x_d,v), \quad x_d \in \mathbb{R}^n, \,\, v \in \mathbb{R}^q
\end{align}
where the external signal $v$ and the initial condition of the signal $x_d$ determine the specific reference trajectory. Let the tracking error be defined as $\tilde{x} \triangleq x - x_d$. In general, a controller for tracking a reference trajectory depends on both the tracking error as well as the reference trajectory. Hence, we assume that the control signal is of the form
\begin{equation}
u = \gamma (\xi), \quad \text{where} \quad \xi \triangleq [\tilde{x}; x_d; v]
\label{eqn:u}
\end{equation}
where the notation $[a_1; a_2; a_3]$ denotes the column vector formed by the concatenation of the vectors $a_1$, $a_2$ and $a_3$.
Consequently, the closed loop system that describes the tracking error is given as
\begin{equation}
\dot{\tilde{x}} = f(\tilde{x}+x_d, \gamma (\xi) ) - \dot{x}_d.
\label{eqn:xt_uc}
\end{equation}

Now, consider a controller that updates the control only intermittently and not continuously in time. Let $t_i$ for $i = 0, 1, 2, \ldots$ be the time instants at which the control is computed and updated. Then, the tracking error evolves as
\begin{align}
\dot{\tilde{x}} = f \big(\tilde{x}+x_d, \gamma(\xi(t_i)) \big) - \dot{x}_d, \quad \text{for } t \in [t_i, t_{i+1}), \text{ } i \in \{0, 1, 2, ...\}.
\label{eqn:xt_ut}
\end{align}
The above dynamical system can also be viewed as a continuously updated control system, albeit with an error in the measurement of the state and the exogenous input. By defining the measurement error as
\begin{align}
&e \triangleq \begin{bmatrix}
     e_{\tilde{x}} \\
     e_{x_d} \\
     e_{v}
\end{bmatrix}
\triangleq \xi(t_i) - \xi \triangleq \begin{bmatrix}
     \tilde{x}(t_i) - \tilde{x} \\
     x_d(t_i) - x_d \\
     v(t_i) - v
\end{bmatrix}, \quad \text{for } t \in [t_i, t_{i+1}), \text{ } i \in \{0, 1, 2, ...\} \label{eqn:edef}
\end{align}
the system in (\ref{eqn:xt_ut}) can be rewritten as
\begin{align}
\dot{\tilde{x}} &= \big[f(\tilde{x}+x_d, \gamma (\xi) ) - \dot{x}_d \big] + \big[ f(\tilde{x}+x_d, \gamma(\xi+e)) - f(\tilde{x}+x_d, \gamma (\xi)) \big]
\label{eqn:perturbed_sys}
\end{align}
where we have expressed the above system as a perturbed version of the dynamical system \eqref{eqn:xt_uc}. Note that $e$ is discontinuous at $t = t_i$, for each $i$, because $e(t_i) = \xi(t_i) - \xi(t_i) = 0$ while $\displaystyle e(t_i^-) \triangleq \lim_{t \uparrow t_i} e(t) = \lim_{t \uparrow t_i} (\xi(t_{i-1}) - \xi(t))$.

In time-triggered or periodic control systems, $t_{i+1} - t_i = T_s$ for all $i \in \{0, 1, 2, \ldots \}$, where $T_s > 0$ is a constant sampling time. On the other hand, in an event-triggered system the time instants $t_i$ in general are not uniformly spaced, and are determined dynamically by an event-triggering condition.

The objective of this paper is to develop an event based controller for tracking a trajectory within a desired ultimate bound. To this end, we assume that when the control is updated continuously in time, the state $x$ tracks the desired trajectory asymptotically, that is, there exists $\gamma$ such that system \eqref{eqn:xt_uc} satisfies $\tilde{x} \rightarrow 0$ as $t \rightarrow \infty$. In the next section, we specify the main assumptions of the paper and develop an event-triggering condition for tracking a given trajectory within a desired bound.


\section{Event-triggering condition for emulation based trajectory tracking control}
\label{sec:trig_con}

There are two main requirements for an event based trajectory tracking controller. It needs to (i) guarantee that the tracking error is at least uniformly ultimately bounded, and (ii) ensure that there is no accumulation of execution times. In this section, an event-triggering condition that satisfies both of these requirements is developed. We begin by formally stating the \textbf{main assumptions} of this paper.
\newcounter{saveenum}
\begin{enumerate}[label={\textbf{(A\arabic*)}},ref={A\arabic*}]
\item\label{A:asymp} Suppose $f(0, \gamma(0)) - f_r(0,0) = 0$ and that there exists a $C^1$ Lyapunov function for the dynamical system in (\ref{eqn:xt_uc}), $V : \mathbb{R}^n \rightarrow \mathbb{R}$, such that for all admissible $x_d$ and $v$,
\begin{align*}
&\alpha_1(\norm{\tilde{x}}) \leq V(\tilde{x}) \leq \alpha_2(\norm{\tilde{x}}) \notag \\
&\frac{\partial V}{\partial \tilde{x}} \big[ f (\tilde{x}+x_d,\gamma(\xi)) - f_r(x_d,v) \big] \leq - \alpha_3(\norm{\tilde{x}})
\end{align*}
where $\alpha_1(.)$, $\alpha_2(.)$, and $\alpha_3(.)$ are class $\mathcal{K}_\infty$ functions\footnote{A continuous function $\alpha : [0, \infty) \rightarrow [0, \infty)$ is said to belong to the class $\mathcal{K}_{\infty}$ if it is strictly increasing, $\alpha(0) = 0$ and $\alpha(r) \rightarrow \infty$ as $r \rightarrow \infty$ \cite{khalil2002_book}.}.
\item\label{A:lipschitz} The functions $f$, $\gamma$ and $f_r$ are Lipschitz on compact sets.
\item\label{A:xd_v_bound} For all time $t \geq 0$, $\norm{[x_d;v]} \leq d$ for some $d \geq 0$ and $v$ is piecewise continuous.
\item\label{A:v} For all time $t \geq 0$, $v$ is differentiable and $\norm{\dot{v}} \leq c$ for some $c \geq 0$.
\setcounter{saveenum}{\value{enumi}}
\end{enumerate}
The notation $\norm{.}$ denotes the Euclidean norm of a vector. In the sequel, it is also used to denote the induced Euclidean norm of a matrix. Note that the meaning of `admissible $x_d$ and $v$' in \eqref{A:asymp} differs in each of our main results, where in each case it is specified precisely. At this stage, it is enough to know that \eqref{A:xd_v_bound} is satisfied in each case. Now, consider the following family of compact sets:
\begin{align}
S(R) = \{ \xi : V(\tilde{x}) \leq \alpha_2(R), \, \norm{[x_d; v]} \leq d \}, \,\,\, \delta S(R) = \{ \xi : V(\tilde{x}) = \alpha_2(R), \, \norm{[x_d; v]} \leq d \}
\label{eqn:Sdef}
\end{align}
Note that for each $R \geq 0$, the sets $S(R)$ and $\delta S(R)$ include all the admissible reference signals, $x_d$ and $v$. For each set $S(R)$ there exists, by assumption (\ref{A:lipschitz}), a vector $L(R) \in \mathbb{R}^{2n+q}$, with each of its components greater than zero such that
\begin{align}
\norm{f(\tilde{x} + x_d, \gamma(\xi+e)) - f(\tilde{x} + x_d,\gamma(\xi)) } \leq L(R)^T |e| \leq \norm{L(R)} \norm{e}, \,\, \forall \, \xi, \, (\xi+e) \in S(R)
\label{eqn:lip_gamma}
\end{align}
where $|e|$ denotes the vector of the absolute values of the components of $e$. Without loss of generality, it may be assumed that each component of $L(R)$ is a non-decreasing function of $R$. In the sequel, we use the notation $S_i$, $\delta S_i$ and $L_i$ to denote $S(\norm{\tilde{x}(t_i)})$, $\delta S(\norm{\tilde{x}(t_i)})$ and $L(\norm{\tilde{x}(t_i)})$, respectively. Next, we define a continuous function, $\beta(.)$, that satisfies
\begin{equation}
\beta(R) \geq \max_{\norm{w} \leq R} \bigg|\bigg|\frac{\partial V(w)}{ \partial w} \bigg|\bigg|, \quad \forall R \geq 0
\label{eqn:beta}
\end{equation}
We now derive the triggering condition that determines the time instants $t_i$ at which the control is updated.

Consider the Lyapunov function, $V(.)$, in assumption (\ref{A:asymp}) as a candidate Lyapunov function for the system \eqref{eqn:xt_ut}. The time derivative of $V(\tilde{x})$, along the flow of the tracking error system, $\dot{V} = (\partial V / \partial \tilde{x})\dot{\tilde{x}}$, may be obtained through the measurement error interpretation, (\ref{eqn:perturbed_sys}).
\begin{align}
\dot{V} =& \frac{\partial V}{\partial \tilde{x}} \big[f(\tilde{x}+x_d,\gamma(\xi)) - \dot{x}_d \big] + \frac{\partial V}{\partial \tilde{x}} \big[f(\tilde{x}+x_d,\gamma(\xi+e)) - f(\tilde{x}+x_d,\gamma(\xi)) \big] \notag \\
\leq & - \alpha_3(\norm{\tilde{x}}) + \frac{\partial V}{\partial \tilde{x}} \big[f(\tilde{x}+x_d,\gamma(\xi+e)) - f(\tilde{x}+x_d,\gamma(\xi)) \big] \label{eqn:Vdot_asymp}\\
\leq & - \alpha_3(\norm{\tilde{x}}) + \beta(\norm{\tilde{x}}) L(R)^T |e|, \quad \forall \, \xi, \, (\xi+e) \in S(R) \label{eqn:Vdot}
\end{align}
where \eqref{eqn:Vdot_asymp} is obtained from assumption (\ref{A:asymp}), and \eqref{eqn:Vdot} is then obtained from \eqref{eqn:Sdef}-\eqref{eqn:beta}. Then, \eqref{eqn:Vdot} suggests a triggering condition.

Consider the following \textbf{triggering condition} (for the sake of clarity, the complete system description including the state equation and the triggering condition are given).
\begin{align}
&\dot{\tilde{x}} = f \big(\tilde{x}+x_d, \gamma(\xi(t_i)) \big) - \dot{x}_d, \quad \text{for } t \in [t_i, t_{i+1}), \text{ } i \in \{0, 1, 2, ...\} \label{eqn:xt_trig}\\
&t_0 = \min \{t \geq 0 : \norm{\tilde{x}} \geq r > 0 \}, \text{ and} \notag\\
&t_{i+1} = \min \{t \geq t_i : L_i^T |\xi(t_i) - \xi| -\frac{\sigma \alpha_3(\norm{\tilde{x}})}{\beta(\norm{\tilde{x}})} \geq 0 \text{ and } \norm{\tilde{x}} \geq r > 0 \} \label{eqn:trig_con}
\end{align}
where $0 < \sigma < 1$ and $r$ is a design parameter that determines the ultimate bound of the tracking error. It is necessary to update the control only when $\norm{\tilde{x}} \geq r$, for some $r > 0$, else it may result in the accumulation of control update times. Notice that each update instant $t_{i+1}$ is defined implicitly with respect to $t_i$. Hence, the initial update instant $t_0$ has been specified separately. As the proposed triggering condition does not allow the control to be updated whenever $\norm{\tilde{x}} < r$, the first update instant, $t_0$, need not be at $t = 0$. Therefore, it is assumed that $u = 0$ for $0 \leq t < t_0$. In the next section the triggering condition \eqref{eqn:trig_con} is shown to guarantee uniform ultimate boundedness of the tracking error for the reference trajectories considered in this paper.

\section{Uniform ultimate boundedness of the trajectory tracking error}
\label{sec:ultim_track}

The following lemma demonstrates, under \eqref{A:asymp}-\eqref{A:xd_v_bound}, that the event-triggering condition \eqref{eqn:trig_con} ensures $\xi \in S_i$ for all $t \in [t_i, t_{i+1})$, for each $i$. Moreover, the lemma also demonstrates that the event-triggering condition \eqref{eqn:trig_con} renders the tracking error ultimately bounded, provided the sequence of control execution times does not exhibit \textit{Zeno behavior} (accumulation of inter-execution times), in other words either the sequence of control execution times is finite or $\displaystyle \lim_{i \rightarrow \infty} t_i = \infty$.

\begin{lemma}\label{lem:bounded}
Consider the system \eqref{eqn:xt_uc}. Suppose that assumptions (\ref{A:asymp}), (\ref{A:lipschitz}) and (\ref{A:xd_v_bound}) are satisfied. Then, in the event-triggered system (\ref{eqn:xt_trig})-\eqref{eqn:trig_con}, for each $i$, $\xi \in S_i$ for all $t \in [t_i, t_{i+1})$. Further, if the initial condition is bounded and the sequence of control execution times does not exhibit Zeno behavior, then the tracking error, $\tilde{x}$, is uniformly ultimately bounded by a ball of radius $r_1 = \alpha_1^{-1}(\alpha_2(r))$. 
\end{lemma}

\begin{proof}
First, we establish by contradiction that for each $i$, $\xi \in S_i$ for all $t \in [t_i, t_{i+1})$. Note that by definition, $(\xi+e) = \xi(t_i) \in S_i$ and the triggering condition enforces $\norm{\tilde{x}(t_i)} \geq r$. Further, since $\norm{\tilde{x}(t_i)} \geq r$, the open $r$-ball is a proper subset of and is contained within the interior of $S_i$ (that is, its intersection with $\delta S_i$ is an empty set). Also note that sets $S_i$ and $\delta S_i$ (see \eqref{eqn:Sdef} and the text following \eqref{eqn:lip_gamma}) are essentially a sub-level set and a level set, respectively, of the Lyapunov function $V$. Now, let us assume that $\xi$ does escape $S_i$ during the interval $[t_i, t_{i+1})$. Then, since the tracking error $\tilde{x}$ is continuous as a function of time, there exists a $t_i^* \in [t_i, t_{i+1})$ such that $\xi(t_i^*) \in \delta S_i \subset S_i$ and $\dot{V}|_{t = t_i^*} > 0$ (where $\dot{V}|_{t = t_i^*}$ denotes $\dot{V}$ evaluated at $t = t_i^*$). However, as $\xi(t_i^*) \in \delta S_i \subset S_i$, \eqref{eqn:Vdot} and \eqref{eqn:trig_con} imply $\dot{V}|_{t = t_i^*} \leq - (1 - \sigma)\alpha_3(\norm{\tilde{x}(t_i^*)}) < 0$. Thus, having arrived at a contradiction, we conclude that no such $t_i^*$ exists and that the first claim of the lemma is true. Consequently, \eqref{eqn:Vdot} and \eqref{eqn:trig_con} again imply that the derivative $\dot{V}$ along the flow of the system satisfies
\begin{align}
\dot{V} &\leq - (1 - \sigma)\alpha_3(\norm{\tilde{x}}) < 0, \quad \forall t \in [t_i, t_{i+1} ), \,\,\, \text{s.t. } \norm{\tilde{x}(t)} \geq r \label{eqn:Vdot_neg}
\end{align}
and further, for each $R \geq r$ it is true that any solution that enters the set $S(R)$ does not leave it subsequently.

The assumption that $\tilde{x}(0)$ is bounded and the definition of $t_0$ imply that $\tilde{x}(t_0)$ is also bounded. Then, the assumption that the sequence of control execution times does not exhibit Zeno behavior implies that the triggering condition, \eqref{eqn:trig_con}, is well defined and that $\dot{V} \leq - (1 - \sigma)\alpha_3(\norm{\tilde{x}}) < 0$, $\forall t \in [0, \infty)$ s.t. $\norm{\tilde{x}(t)} \geq r$ (if there are finitely many control updates, that is $i \in \{ 0, 1, \ldots, N \}$, then $t_{N+1} = \infty$). Then, in fact, it is true that $S(R)$ is positively invariant for each $R \geq r$. In particular, $S_0$ is positively invariant. Then, \eqref{eqn:Vdot_neg} implies that $\dot{V} \leq - (1 - \sigma)\alpha_3(r) < 0$ for all $\xi \in S_0$ such that $\norm{\tilde{x}} \geq r$. Hence all solutions, $\xi$, with bounded initial conditions enter the set $S(r)$ in finite time and as $S(r)$ is positively invariant, the solutions stay there. Therefore the tracking error, $\tilde{x}$, is uniformly ultimately bounded by the closed ball of radius $r_1 = \alpha_1^{-1}(\alpha_2(r))$.
\end{proof}

Looking back at \eqref{eqn:trig_con}, it is clear that the functions $\alpha_3$ and $\beta$ play a crucial role in determining how often an event is triggered or in computing a lower bound for the inter-execution times. Specifically, the following definition is used in the sequel.
\begin{equation}
\Delta_{s_1}^{s_2} \triangleq \min_{s_1 \leq \norm{\tilde{x}} \leq s_2} \sigma \alpha_3(\norm{\tilde{x}})/\beta(\norm{\tilde{x}}) \label{eqn:Delta}
\end{equation}
where $s_2 \geq s_1 > 0$ are any positive real numbers, the functions $\alpha_3$ and $\beta$ are as defined in \eqref{A:asymp} and \eqref{eqn:beta}, respectively. Since $\alpha_3$ and $\beta$ are continuous positive definite functions, $\Delta_{s_1}^{s_2}$ is well defined and positive for any given $s_2 \geq s_1 > 0$.

Now we present the first main result of the paper. It demonstrates, for a particular class of reference trajectories, that in the event-triggered system \eqref{eqn:xt_trig}-\eqref{eqn:trig_con} the inter-execution times are uniformly bounded away from zero while the tracking error is uniformly ultimately bounded.

\begin{theorem}
\label{thm:vdot_bound}
Consider the system (\ref{eqn:xt_uc}). Suppose that assumptions (\ref{A:asymp}), (\ref{A:lipschitz}), (\ref{A:xd_v_bound}) and (\ref{A:v}) are satisfied. Then, for the event-triggered system (\ref{eqn:xt_trig})-\eqref{eqn:trig_con}, the tracking error, $\tilde{x}$, is uniformly ultimately bounded by a ball of radius $r_1 = \alpha_1^{-1}(\alpha_2(r))$, and the inter-execution times ($t_{i+1} - t_i$) for $i \in \{0, 1, 2, \ldots \}$ are uniformly bounded below by a positive constant that depends on the bound of the initial tracking error.
\end{theorem}

\begin{proof}
Uniform ultimate boundedness of the tracking error follows from Lemma \ref{lem:bounded}. Only the existence of a positive lower bound for the inter-execution times remains to be shown. Note that for each $i$, $\norm{e(t_i)} = 0$ and $\norm{\tilde{x}(t_i)} \geq r$. Hence, the triggering condition \eqref{eqn:trig_con} implies that the $i^{\text{th}}$ inter-update time, $(t_{i+1} - t_i)$, is at least equal to the time it takes $\norm{L_i} \norm{e}$ to grow from $0$ to $\sigma \alpha_3(\norm{\tilde{x}}) / \beta(\norm{\tilde{x}})$. Recall from the proof of Lemma \ref{lem:bounded} that every solution, $\xi$, stays in the set $S_0$ for all $t \in [t_0, t_i)$, for each $i$. Thus, $\norm{L_i} \leq \norm{L_0}$ for each $i$. Notice that
\begin{equation}
S_0 \subset \{ \xi : \norm{\tilde{x}} \leq \mu_0, \, \norm{[x_d; v]} \leq d \} \label{eqn:S0_bound}
\end{equation}
where ${\mu_0} = \alpha_1^{-1}(\alpha_2(\norm{\tilde{x}(t_0)}))$. Then, \eqref{eqn:Delta} implies $t_{i+1} - t_i \geq T$, where $T$ is the time it takes $\norm{e}$ to grow from $0$ to $\Delta_r^{\mu_0} / \norm{L_0}$. If we show that $T > 0$, then the proof is complete.

From \eqref{eqn:perturbed_sys}, and the triangle inequality property, we observe that
\begin{align}
\norm{\dot{\tilde{x}}} \leq \norm{f(\tilde{x}+x_d, \gamma (\xi) ) - \dot{x}_d} + \norm{ f(\tilde{x}+x_d, \gamma(\xi+e)) - f(\tilde{x}+x_d, \gamma (\xi))}
\end{align}
From \eqref{eqn:lip_gamma}, the second term is bounded by $L_0^T |e| \leq \norm{L_0} \norm{e}$ on the set $S_0$. Since, according to (\ref{A:asymp}), $f(0, \gamma (0) ) - f_r(0,0) = 0$, \eqref{A:lipschitz} then implies that there exist Lipschitz constants $P_1 \geq 0$ and $P_2 \geq 0$ such that
\begin{align*}
\norm{\dot{\tilde{x}}} &\leq P_1 \norm{\tilde{x}} + P_2 \norm{[x_d; v]} + L_0^T |e|\\
&\leq P_1 {\mu_0} + P_2 d + \norm{L_0} \norm{e}
\end{align*}
where the second inequality is obtained from \eqref{eqn:S0_bound}. Assumptions \eqref{A:lipschitz}-\eqref{A:xd_v_bound} imply that there exists a constant $P_3 \geq 0$ such that $\norm{\dot{x}_d} \leq P_3 d$ and \eqref{A:v} implies $\norm{\dot{v}} \leq c$. Then, by letting $P_0 = P_1 {\mu_0} + (P_2 + P_3) d$ and from the definition $\dot{e} = - [\dot{\tilde{x}}; \dot{x}_d; \dot{v}]$ it follows that
\begin{equation}
\frac{\mathrm{d}\norm{e}}{\mathrm{d}t} \leq \norm{\dot{e}} \leq \norm{L_0} \norm{e} + P_0 + c
\end{equation}
Note that for $\norm{e} = 0$, the first inequality holds for all the directional derivatives of $\norm{e}$. Then, according to the Comparison Lemma \cite{khalil2002_book}
\begin{equation}
\norm{e} \leq \frac{P_0 + c}{\norm{L_0}} (\mathrm{e}^{\norm{L_0}(t-t_i)} - 1), \quad \text{for} \quad t \geq t_i.
\end{equation}
Thus, the inter-execution times are uniformly lower bounded by $T$, which satisfies
\begin{align}
T \geq \frac{1}{\norm{L_0}} \log \bigg(1 + \frac{\Delta_r^{\mu_0}}{P_0 + c} \bigg).
\label{eqn:T_case1}
\end{align}
As $\norm{L_0}$ is finite and $\Delta_r^{\mu_0} > 0$, we conclude that the inter-execution times have a uniform lower bound, $T$, that is greater than zero.
\end{proof}


In the next result, the conditions on the reference trajectory are relaxed by no longer requiring it to satisfy assumption (\ref{A:v}). Instead, to ensure the absence of Zeno behavior, a new assumption is made - that $d_v$, the uniform bound on $\norm{v}$, is no larger than a quantity determined by $\Delta_r^{\mu_0}$ and $L_0$. The new assumptions, in contrast to Theorem \ref{thm:vdot_bound}, lead to a constraint on the choice of the radius $r$ in the triggering condition and ensure only local uniform ultimate boundedness of the trajectory tracking error. Let $L(R) \triangleq [ Q(R); M(R) ]$ and $L_i \triangleq [ Q_i; M_i ]$ where $Q(R), Q_i \in \mathbb{R}^{2n}$ and $M(R), M_i \in \mathbb{R}^q$. Now, the second main result is presented.

\begin{theorem}
\label{thm:v_bound}
Consider the system defined by (\ref{eqn:xt_uc}). Suppose that the assumptions (\ref{A:asymp}), (\ref{A:lipschitz}) and (\ref{A:xd_v_bound}) hold. Also, for some $R_0 \geq r$ suppose that $\Delta_r^{\mu_0} - 2d_v \norm{M(R_0)} > 0$, where ${\mu_0} = \alpha_1^{-1}(\alpha_2(R_0))$, $\Delta_r^{\mu_0}$ is given by \eqref{eqn:Delta} and $d_v$ is the uniform bound on $\norm{v}$. If $\norm{\tilde{x}(0)} \leq R_0$, then in the event-triggered system (\ref{eqn:xt_trig})-\eqref{eqn:trig_con}, the tracking error, $\tilde{x}$, is uniformly ultimately bounded by a ball of radius $r_1 = \alpha_1^{-1}(\alpha_2(r))$, and the inter-update times $(t_{i+1} - t_i)$ for $i \in \{0, 1, 2, \ldots \}$ are uniformly bounded below by a positive constant that depends on $R_0$.
\end{theorem}

\begin{proof}
The proof is very similar to that of Theorem \ref{thm:vdot_bound}, and hence only the essential steps are described here. According to Lemma \ref{lem:bounded} each solution, $\xi$, with $\norm{\tilde{x}(0)} \leq R_0$ stays in the set $S(R_0)$. Hence, $\norm{M_i} \leq \norm{M(R_0)}$ and $\norm{Q_i} \leq \norm{Q(R_0)}$ for each $i$. Since $\norm{v}$ is uniformly bounded by $d_v$ it follows that for each $i$, $M_i^T |e_v| \leq \norm{M_i} \norm{e_v} \leq 2d_v \norm{M(R_0)}$, where $e_v = v(t_i) - v$ and $|e_v|$ denotes the component-wise absolute value of the vector $e_v$. The definitions of $Q_i$ and $M_i$ imply that $L_i^T |e| = Q_i^T |[e_{\tilde{x}}; e_{x_d}]| + M_i^T |e_v| \leq Q_i^T |[e_{\tilde{x}}; e_{x_d}]| + 2d_v \norm{M(R_0)}$.

Note that for each $i$, $r \leq \norm{\tilde{x}(t_i)} \leq {\mu_0}$. Thus, the triggering condition in (\ref{eqn:trig_con}) implies that for each $i$, $L_{i-1}^T |e(t_i^-)| \geq \Delta_r^{\mu_0}$, or equivalently, $Q_{i-1}^T |[e_{\tilde{x}}(t_i^-); e_{x_d}(t_i^-)]| \geq \delta \triangleq \Delta_r^{\mu_0} - 2d_v \norm{M(R_0)} > 0$, the last inequality being one of the assumptions. Hence, the inter-execution times $t_{i+1} - t_i \geq T$, where $T$ is the time it takes $\norm{[e_{\tilde{x}}; e_{x_d}]}$ to grow from $0$ to $\delta /\norm{Q(R_0)}$. If we show that $T > 0$, then the proof is complete.

Following steps similar to those in the proof of Theorem \ref{thm:vdot_bound}, we know that there exists a finite $P_0 \geq 0$ such that $\displaystyle \frac{\mathrm{d}\norm{[e_{\tilde{x}}; e_{x_d}]}}{\mathrm{d}t} \leq \norm{Q_0} \norm{[e_{\tilde{x}}; e_{x_d}]} + P_0 + 2d_v \norm{M(R_0)}$. Note that for $\norm{[e_{\tilde{x}}; e_{x_d}]} = 0$, the inequality holds for all the directional derivatives. Thus, the inter-execution times are uniformly lower bounded by $T$, which satisfies
\begin{align}
T \geq \frac{1}{\norm{Q_0}} \log \bigg(1 + \frac{\Delta_r^{\mu_0} - 2d_v}{P_0 + 2d_v \norm{M(R_0)}} \bigg).
\label{eqn:T_case2}
\end{align}
As $\norm{Q_0}$ is finite, we conclude that the inter-execution times have a lower bound, $T$, that is greater than zero.
\end{proof}


Theorem \ref{thm:v_bound} is somewhat conservative because only the uniform bound on $\norm{v}$ is utilized in determining the ultimate bound and the lower bound on the inter-execution times. A more useful result is obtained by imposing only slightly stricter constraints on $v$ - that jumps in $v$ are separated in time by $T_v > 0$, that the magnitude of each jump is upper bounded by a known constant and that $v$ is Lipschitz between jumps. This is expressed formally in the following assumption.
\begin{enumerate}[label={\textbf{(A\arabic*)}},ref={A\arabic*}]
\setcounter{enumi}{\value{saveenum}}
\item \label{A:v_jmp_dwell} There exist constants $c \geq 0$, $T_v \geq 0$ and $J_v \geq 0$ such that for all $t,s \geq 0$, the following holds: $\norm{v(t) - v(s)} \leq c |t -s| + \Big\lceil \frac{|t-s|}{T_v} \Big\rceil J_v$, where $\lceil . \rceil$ is the ceiling function.
\end{enumerate}
Now, the final result is presented.

\begin{theorem}
\label{thm:v_jump_dwell}
Consider the system defined by (\ref{eqn:xt_uc}). Suppose that the assumptions \eqref{A:asymp}, \eqref{A:lipschitz}, \eqref{A:xd_v_bound} and \eqref{A:v_jmp_dwell} hold. Also, for some $R_0 \geq r$ suppose that $\Delta_r^{\mu_0} - J_v \norm{M(R_0)} > 0$, where ${\mu_0} = \alpha_1^{-1}(\alpha_2(R_0))$ and $\Delta_r^{\mu_0}$ is given by \eqref{eqn:Delta}. If $\norm{\tilde{x}(0)} \leq R_0$, then in the event-triggered system (\ref{eqn:xt_trig})-\eqref{eqn:trig_con}, the tracking error, $\tilde{x}$, is uniformly ultimately bounded by a ball of radius $r_1 = \alpha_1^{-1}(\alpha_2(r))$, and the inter-update times $(t_{i+1} - t_i)$ for $i \in \{0, 1, 2, \ldots \}$ are uniformly bounded below by a positive constant that depends on $R_0$.
\end{theorem}

\begin{proof}
Let $e^* \triangleq [e_{\tilde{x}}; e_{x_d}; e_{v^*}]$, where  $e_{v^*} \triangleq c (t - t_i)$ for $t \in [t_i, t_{i+1})$ and each $i$. Then, by \eqref{A:v_jmp_dwell}, $\norm{e} \leq \norm{e^*} + \Big\lceil \frac{|t-s|}{T_v} \Big\rceil J_v$. Now, let $T_k$ be the time it takes $\norm{e^*}$ to grow from zero to $(\Delta_r^\mu - k J_v \norm{M(R_0)})/\norm{L_0}$. Then, a lower bound on the inter-execution times is given by
\begin{equation*}
\max_{k \in \{1,2,\ldots,N\}} \{ \min\{k T_v, T_k\} \}, \quad N = \bigg\lfloor \frac{\Delta_r^\mu}{J_v \norm{M(R_0)}} \bigg\rfloor
\end{equation*}
where $\lfloor . \rfloor$ denotes the floor function. Note that $\{kT_v\}$ is an increasing sequence while $\{T_k\}$ is a decreasing sequence. Following the proof of Theorem \ref{thm:vdot_bound}, $T_k$ is estimated as
\begin{align}
T_k \geq \frac{1}{\norm{L_0}} \log \bigg(1 + \frac{\Delta_r^\mu - k J_v \norm{M(R_0)}}{P_0 + c} \bigg).
\label{eqn:T_case3}
\end{align}
Thus the inter-execution times are lower bounded by a positive constant. The ultimate boundedness of the tracking error follows from Lemma \ref{lem:bounded}.
\end{proof}


\begin{remark}
\label{rem:Li}
Notice from \eqref{eqn:lip_gamma} that in order to compute $L_i = L(\norm{\tilde{x}(t_i)})$ it is necessary to compute the set $S_i = S(\norm{\tilde{x}(t_i)})$ or at least a set of which $S_i$ is a subset, such as $B_i \triangleq \{ \xi: \norm{\tilde{x}} \leq \alpha_1^{-1}(\alpha_2(\norm{\tilde{x}(t_i})), \norm{[x_d; v]} \leq d \}$. However, if $\norm{\tilde{x}(t_i)} \geq \norm{\tilde{x}(t_{i-1})}$ then clearly some components of $L_i$ may be greater than those of $L_{i-1}$. But from Lemma \ref{lem:bounded}, we know that $S_{i} \subset S_{i-1}$ for each $i$, so at time instant $t_i$ instead of computing $L_i$ based on $B_i$, we can let $L_i = L_{i-1}$. Following this rule, the sequence $\{L_i\}$ can be chosen to be component-wise non-increasing. The triggering condition and the estimates of lower bounds on the inter-update times depend critically on $L$ and hence using a time-varying $L$ lowers the overall average update rate. Computing $L$ is in general a computationally costly task and it is not useful to update $L$ continuously in time like $\alpha_3(\norm{\tilde{x}})$ and $\beta(\norm{\tilde{x}})$.
\end{remark}

In the next section our theoretical results are illustrated through simulations of a second order nonlinear system.


\section{Examples and simulation results}
\label{sec:sim}

The theoretical results developed in the previous sections are illustrated through simulations of the following second order nonlinear system.
\begin{align}
\dot{x} =
\begin{bmatrix}
\dot{x}_1\\
\dot{x}_2
\end{bmatrix} =
\begin{bmatrix}
0 & 1\\
0 & -1
\end{bmatrix} x + \begin{bmatrix} 0\\ -x_1^3 \end{bmatrix} + \begin{bmatrix} 0\\ 1 \end{bmatrix} u =
Ax + \begin{bmatrix} 0\\ -x_1^3 \end{bmatrix} + Bu
\label{eqn:nonlin_spring}
\end{align}
The desired trajectory is a solution of the system $[\dot{x}_{d,1}; \dot{x}_{d,2}] = [x_{d,2}; v]$, where $v$ is an exogenous input, which along with the initial conditions of the state of the reference system, $x_d = [x_{d,1}; x_{d,2}]$, determines the specific trajectory. The control function is chosen as 
\begin{align}
\gamma(\xi) = K \tilde{x} + v + (\tilde{x}_1 + x_{d,1})^3 + x_{d,2}
\end{align}
where $K = [k_1; k_2]^T$ is a $2 \times 1$ row vector such that $\tilde{A} = (A+BK)$ is Hurwitz, and $\tilde{x} = [\tilde{x}_1; \tilde{x}_2]$ is the tracking error. Then, the closed-loop tracking error system with event-triggered control can be written, using the measurement error interpretation, as
\begin{align}
&\dot{\tilde{x}}_1 = \tilde{x}_2 \notag \\
&\dot{\tilde{x}}_2 = - (\tilde{x}_2 + x_{d,2}) - (\tilde{x}_1 + x_{d,1})^3 + \gamma(\xi + e) - v.
\label{eqn:xt_nonlin_spring}
\end{align}
Now, consider the quadratic Lyapunov function $V = \tilde{x}^T P \tilde{x}$ where $P$ is a positive definite matrix that satisfies the Lyapunov equation $P\tilde{A} + \tilde{A}^T P = - H$, where $H$ is a given positive definite matrix. The time derivative of $V$ along the flow defined by \eqref{eqn:xt_nonlin_spring} can be shown to satisfy
\begin{align}
\dot{V} &\leq - \tilde{x}^T H \tilde{x} + 2 \tilde{x}^T P B [\gamma(\xi + e) - \gamma(\xi)] \notag\\
&\leq - \sigma a \norm{\tilde{x}}^2 + \beta(\norm{\tilde{x}}) L(R)^T |e|, \quad \forall \xi, (\xi+e) \in S(R)
\end{align}
where $a > 0$ is the minimum eigenvalue of $H$, $\beta(\norm{\tilde{x}}) = 2 \norm{PB} \norm{\tilde{x}}$ and
\begin{equation}
L(R) = \big[ 3(\mu+d_1)^2+|k_1|; |k_2|; 3(\mu+d_1)^2; 1; 1 \big]
\end{equation}
where $\mu = \alpha_1^{-1}(\alpha_2(R))$ and $d_1 \leq d$ is the uniform bound on $x_{d,1}$. If $d_1$ is not known explicitly then $d$ from assumption \eqref{A:xd_v_bound} may be used instead. Note that $B$ has been absorbed in $\beta$ rather than in $L(R)$, as it should have been according to their definitions. This makes the $\beta$ function point-wise lower. The vectors $L_i$ were computed according to the procedure in Remark \ref{rem:Li}. Finally, given a desired ultimate bound for the trajectory tracking error, the parameter $r$ in the triggering condition can be designed. Next, we present simulation results for two cases corresponding to the two main classes of reference trajectories considered in this paper.

\textbf{Case I:} The signals $x_{d,1}$, $x_{d,2}$, and $v$ were chosen as sinusoidal signals with peak-to-peak amplitude $2$. This was done by choosing $[x_{d,1}(0), x_{d,2}(0); v(0)] = [\pi/3; 1; 0]$ and $\dot{v} = - \cos(t)$. The initial condition of the plant was $[x_1(0); x_2(0)] = [5; -1]$. The parameter $d_1$ was chosen as $2.5$ while the actual uniform bounds on $x_{d,1}$ and $\norm{[x_d; v]}$ were observed to be around $2$ and $2.28$, respectively. The parameters in the controller were chosen as $K = -[20; 20]^T$, $\sigma = 0.95$ and $H$ was chosen as the identity matrix. According to Theorem \ref{thm:vdot_bound}, we chose $r = 0.0154$ in the triggering condition to achieve an ultimate bound of $r_1 = 0.1$ in the tracking error.

The simulation results are shown in Figure \ref{fig:nonlin_spring}. The  Figure shows the norm of the tracking error, the radius $r$ in the triggering condition, the desired ultimate bound $r_1$ and $W_i^T|e|$, where $W_i =  (2 \norm{PB} L_i)/(\sigma a)$. The figure demonstrates that the tracking error is ultimately bounded, and well below the desired bound. We recall that according to the triggering condition \eqref{eqn:trig_con}, the control is not updated when $\norm{\tilde{x}} < r$. Hence, as long as $\norm{\tilde{x}} \geq r$, the weighted measurement error, $W_i^T|e|$, is bounded above by the norm of the tracking error, $\norm{\tilde{x}}$, and an event is triggered (control is updated) each time $W_i^T|e| \geq \norm{\tilde{x}}$. However, when $\norm{\tilde{x}} < r$, $W_i^T|e|$ may exceed $\norm{\tilde{x}}$. A zoomed version of the plot in Figure \ref{fig:nonlin_spring} is shown in Figure \ref{fig:nonlin_spring_scale}, where it is clearly seen that the tracking error is only ultimately bounded.
\begin{figure}[!htb]
\centering
\subfloat[Case I]{\label{fig:nonlin_spring}\includegraphics[width=0.33\textwidth]{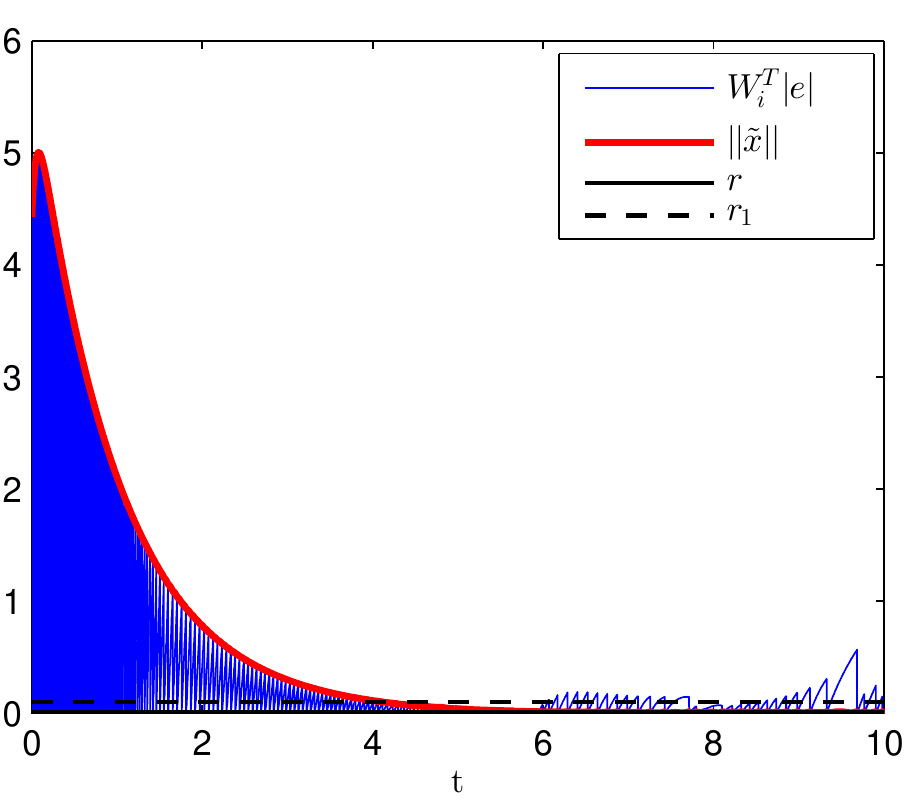}}
\subfloat[Case II]{\label{fig:vquant_nlin_spring}\includegraphics[width=0.33\textwidth]{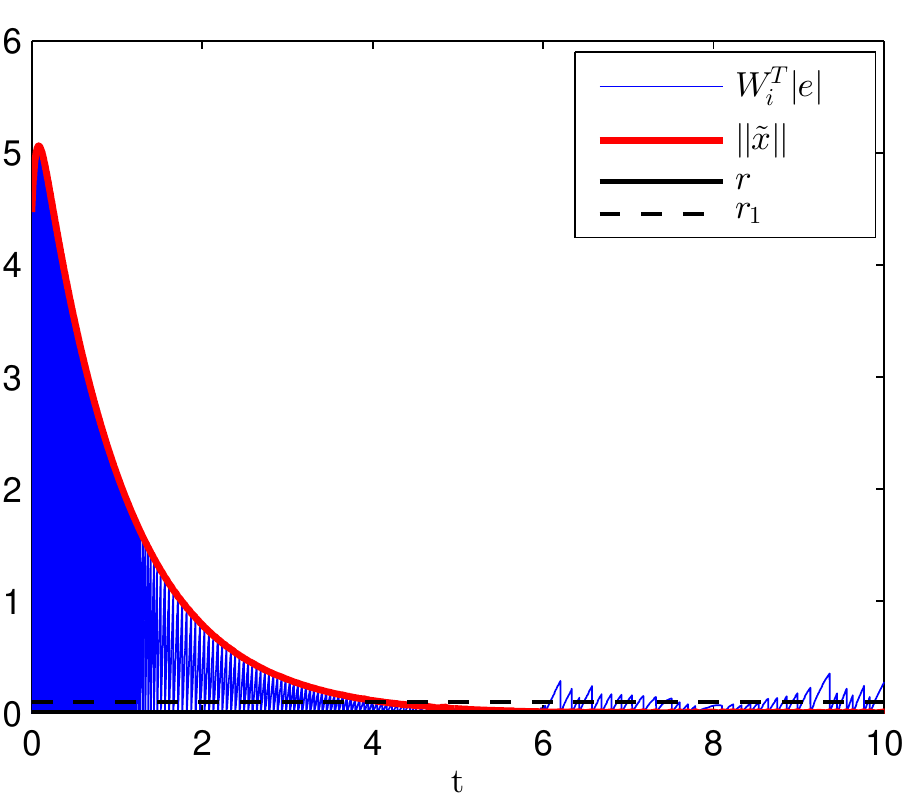}}
\subfloat[Case I (zoom)]{\label{fig:nonlin_spring_scale}\includegraphics[width=0.33\textwidth,height=5.25cm]{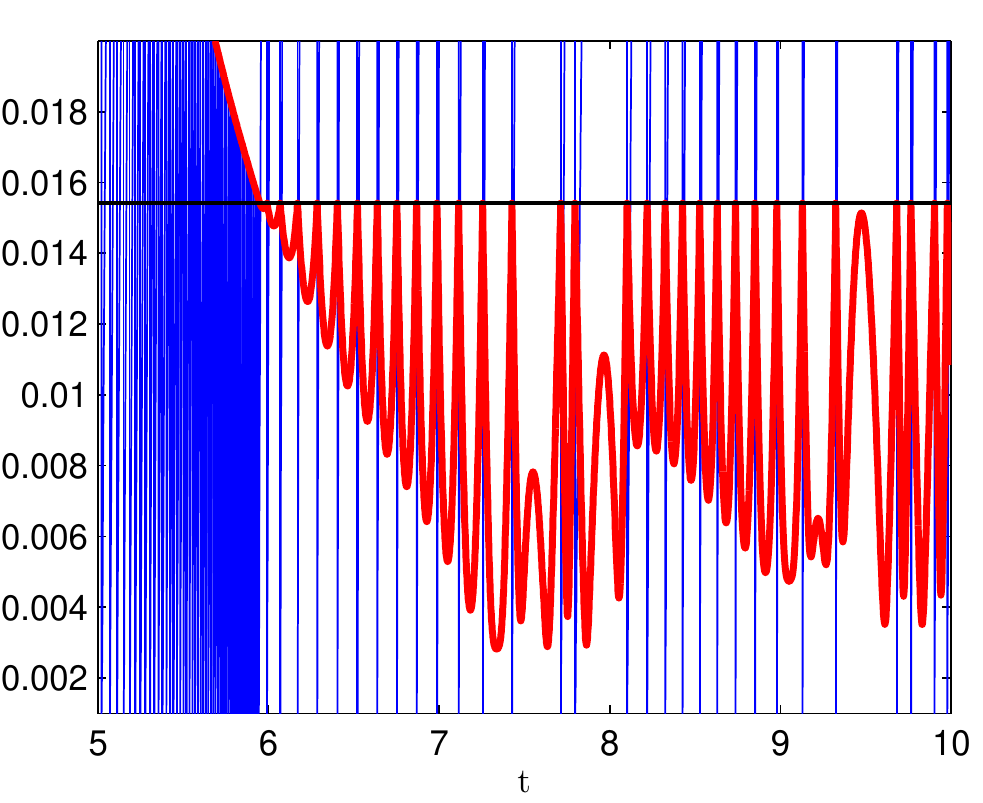}}
\caption{Simulation results for the two reference trajectories.}
\end{figure}

The number of control executions in the simulated time duration was $301$, and the minimum inter-execution time was observed to be $0.005 \text{s}$. The observed average frequency of control updates was around $30\text{Hz}$. Since most of the updates occur before $\tilde{x}$ first enters the ball of radius $r$, it is important to also consider the average frequency for this time period, and in this simulation it was found to be around $46\text{Hz}$. If $L$ is kept constant then these average frequencies are much higher at $943$Hz and $1586$Hz, respectively, with almost no change in the rate of convergence. The theoretical estimate of the minimum inter-execution time is around $6 \times 10^{-8} \text{s}$, which is orders of magnitude lower than the observed value.

\textbf{Case II:} In this case the result in Theorem \ref{thm:v_jump_dwell} is illustrated, where the input signal $v$ is piecewise continuous. In the simulations it was defined as the piecewise constant function, taking values in the set $\mathcal{Q} = \{ 0, \pm 0.1, \pm 0.2, \ldots \}$ and defined as
\begin{equation*}
v(t) = \argmin_{k \in \mathcal{Q}} \{ |- \sin(t) - k|\}
\end{equation*}
For the time instants when $(- \sin(t))$ equals an odd multiple of $0.05$, $v(t)$ is chosen as the higher or the lower of the two possible values based on whether the time derivative of $(- \sin(t))$ is positive or negative, respectively. In the context of Theorem \ref{thm:v_jump_dwell}, the constants $c = 0$ and $J_v = 0.1$.

The initial condition of the reference system was $[x_{d,1}(0); x_{d,2}(0); v(0)] = [1; 1.003; 0]$. From Theorem \ref{thm:v_jump_dwell}, we know that $\Delta_r^\mu$ has to be greater than $J_v = 0.1$, which implies that $r$ has to be greater than $0.0075$. For the example system here, $R_0$ in Theorem \ref{thm:v_jump_dwell} can assume any value. Thus, as in CASE I, $r = 0.0154$ was chosen. The rest of the parameters were the same as in Case I. Figure \ref{fig:vquant_nlin_spring} shows the simulation results. The number of control updates were observed to be $304$, with the minimum execution time at around $0.005$s. The observed average frequencies of control updates were found to be around $30\text{Hz}$ and $46\text{Hz}$ for the simulated time duration and the time duration that $\tilde{x}$ takes to first enter the ball of radius $r$, respectively. These average frequencies are comparable to those in Case I. The theoretical estimate of the minimum inter-execution time is around $3 \times 10^{-8} \text{s}$, which again is very conservative.

\section{Conclusions}
\label{sec:conc}

In this paper, we developed an event based control algorithm for trajectory tracking in nonlinear systems. Using three main results, it was demonstrated that given a nonlinear dynamical system, and a continuous-time controller that ensures uniform asymptotic tracking of the desired trajectory, an event based controller can be designed that not only guarantees uniform ultimate boundedness of the tracking error, but also ensures that the inter-execution times for the control algorithm are uniformly bounded away from zero. The first result demonstrated that uniform boundedness with an arbitrary ultimate bound for the tracking error can be achieved, provided the reference trajectory, the exogenous input to the reference system, and its derivative are all uniformly bounded. However, the minimum \textit{guaranteed} inter-execution time decreases along with the ultimate bound. In the second and third results, we relaxed the assumption on the derivative of the input to the reference system, and demonstrated that the tracking error is uniformly ultimately bounded. In these cases, the analytical results show that it may not be feasible to reduce the ultimate bound below a certain threshold and moreover, the result is only local in general.

The theoretical results were demonstrated through simulations of a second order nonlinear system. The theoretical lower bounds on inter-update times have been found to be very conservative. The reason for this is partially due to the fact that the estimates are based on the rate of change of $\norm{e}$ (made necessary by the presence of exogenous signals) rather than that of $\norm{\tilde{x}}/\norm{e}$ as in \cite{tabuada2007}. Thus, there is significant room for improvement in these estimates and how they are computed. Numerical simulations indicated that the ultimate bound on the tracking error is much lower than the desired value, which is another area for improvement of the theoretical predictions. Finally, it is important to extend these results to output feedback systems.

\section{Acknowledgements}
The authors thank anonymous readers for their helpful comments.


\bibliographystyle{IEEEtran}
\bibliography{IEEEabrv,../Bib/control_refs}

\end{document}